\documentclass{article}
\usepackage{todonotes}
\usepackage{comment}

\usepackage{float}
\usepackage{graphicx}      
\usepackage{subcaption}
\usepackage{color}
\usepackage{siunitx}
\usepackage{amsmath,amssymb}
\usepackage{algorithm}
\usepackage[noend]{algpseudocode}
	
\usepackage{xcolor}

\title{\LARGE \bf
	Reference Governor for Input-Constrained MPC to Enforce State Constraints at Lower Computational Cost
}

\author{Miguel Castroviejo-Fernandez$^{1}$, Jordan Leung$^{1}$ and Ilya Kolmanovsky$^{1}$
\thanks{$^{1}$University of Michigan, Ann Arbor, MI 48109 USA 
        {\tt\small mcastrov, jmleung, ilya@umich.edu}. This research is supported by 
        Air Force Office of Scientific Research Grant number FA9550-20-1-0385.}%
}
\newtheorem{ass}{Assumption}

\newtheorem{lem}{Lemma}
\newtheorem{thm}{Theorem}
\newtheorem{rmk}{Remark}

\newenvironment{proof}{\textbf{Proof:}}{\hfill$\blacksquare$}
\renewcommand{\d}[1]{\ensuremath{\operatorname{d}\!{#1}}}

\begin{document}


\maketitle
\thispagestyle{empty}
\pagestyle{empty}
\begin{abstract}
In this paper, a control scheme is developed based on an input constrained Model Predictive Controller (MPC) and the idea of modifying the reference command to enforce constraints, usual of Reference Governors (RG). The proposed scheme, referred to as the RGMPC, requires optimization for MPC with input constraints for which fast algorithms exist, and can handle (possibly nonlinear) state and input constraints. Conditions are given that ensure recursive feasibility of the RGMPC scheme and finite-time convergence of the modified command to the the desired reference command. Simulation results for a spacecraft rendezvous maneuver with linear and nonlinear constraints demonstrate that the RGMPC scheme has lower average computational time as compared to state and input constrained MPC with similar performance.
\end{abstract}

\section{introduction}

Model Predictive Control (MPC) is informed by optimization of a state and input dependent cost function. At each time step, the input sequence that minimizes this cost subject to constraints on the inputs and/or the states \cite{rawlings2017model} is computed and the input is set to the first element of the sequence. While MPC has emerged as an effective control strategy for constrained systems and is used in many applications, one of its primary drawbacks is the high computational cost associated with solving the optimization problem at each time step.  This computational cost can be significantly lowered in the case of short horizon Linear Quadratic MPC (LQ MPC) with only input constraints by exploiting the underlying structure of the cost to speed up gradient computations as in the Fast MPC algorithm  of \cite{kogel2011fast} or by employing accelerated primal projected gradient methods \cite{nesterov1983method}. In addition, it is easier to enforce anytime feasibility properties \cite{hewing2019scenario} for input constrained MPC (e.g., by saturating the computed input in the case of boxed constraints),  analyze the impact of inexact implementation \cite{liao2020time,leung2021computable}, certify an inexact solution \cite{richter2011computational} and exploit the regularity properties as compared to the state constrained case.  For example, \cite{skibik2022analysis} performs the analysis of an inexact implementation of state and input constrained MPC. Finally, to handle nonlinear constraints the use of more computationally expensive nonlinear MPC is required.

To capitalize on advantages of short-horizon input constrained MPC (uMPC) with polytopic input constraints  yet be able to handle state constraints and (possibly nonlinear) input constraints, in this paper we consider the augmentation of uMPC with a reference governor (RG). RGs \cite{garone2017reference} are add-on schemes that ensure, at each time step, selection of the reference command so that subsequent trajectories remain feasible with respect to constraints. However, the direct application of existing RGs to uMPC-based closed-loop systems is difficult.  For instance, if RG is based on online prediction \cite{bemporad1998reference,tsourapas2009incremental}, a uMPC optimization problem will need to be solved at each time step over the reference governor prediction horizon; this will likely exceed the computational cost of a state and input constrained MPC (cMPC).

In this paper we propose a new scheme which enables a computationally efficient application of RGs to complement uMPC in controlling linear systems with (possibly nonlinear) state constraints and nonlinear input constraints. This scheme, that we refer to as RGMPC, only requires that a single uMPC optimization problem be solved per time step.

For the proposed RGMPC scheme we show, under suitable assumptions, the recursive feasibility as well as finite-time convergence of the modified reference command to the desired constant reference command, i.e. properties expected of conventional RGs. Simulation results for a spacecraft rendezvous (RdV) problem demonstrate low computational requirements and good closed-loop performance being achieved with the proposed approach.

The paper is organized as follows. In Section \ref{sec:theo} the class of systems being addressed is discussed and the two main ingredients: uMPC and the Incremental Reference Governor (IRG) of \cite{tsourapas2009incremental}, needed for subsequent developments are reviewed.
Section \ref{sec:nIRG} introduces the proposed RGMPC scheme and presents theoretical results.
Finally, numerical simulations of the proposed scheme applied to a spacecraft RdV maneuver are reported in Section \ref{sec:EXAMPLE}.

\textbf{Notations:} $\mathbb{S}^n_{++}$, $\mathbb{S}^n_{+}$ denote the set of  symmetric $n\times n$ positive definite and positive  semi-definite matrices respectively. $I_m$ denotes the $m\times m$ identity matrix.
Given $x\in \mathbb{R}^n$ and $W\in \mathbb{S}^n_+$, the W-norm of $x$ is $||x||_W = \sqrt{x^{\top}Wx}$.
Given $P\in\mathbb{S}^n_{++},\;y\in\mathbb{R}^n$, $\mathcal{B}_P(y,r) = \left\{x\in \mathbb{R}^n\;|\; ||y-x||_P \leq r\right\}$ and $\lambda_+(P)$ is the maximum eigenvalue of $P$. Given $a\in\mathbb R^n,\;b\in\mathbb R^m,\; (a,b) = [a^\top,b^\top]^\top$. The sequence made of the $\alpha_j\in\mathbb R^n,\; j=a,\dots,\;b$ elements is denoted by $\{\alpha_j\}_{j=a}^b$.The set $\mathbb{N}$
is the set of positive integers and $\mathbb{N}_0$ the set of non negative ones.

\section{Preliminaries}
\label{sec:theo}
\subsection{Class of systems}
We consider a class of systems represented by the following linear discrete-time models,
\begin{subequations}
\label{eq:IRG_dyn}
\begin{align}
  &x_{k+1} = A x_k + B u_k,\label{eq:IRG_dyn1}\\
    &y_k = C x_k  \label{eq:IRG_dyn2},
 \end{align}
\end{subequations}
    where $A\in\mathbb R^{n\times n}$, $B\in\mathbb R^{n\times m}$, $C\in\mathbb R^{p\times n}$ and $k\in\mathbb N_0$.
    The system is subject to hard constraints on both states and inputs:
    \begin{subequations}
    \label{eq:IRG_cstrSetTots}    
        \begin{align}
    z_k = &(x_k,\;u_k) \in \mathcal{Z},\quad \forall k\geq 0, \label{eq:IRG_cstrSet}\\
    \mathcal{Z} = &\left\{(x,\;u)\;|\;  x\in \mathcal X, \quad u \in \mathcal U \right\}\subseteq\mathbb{R}^{n+m},
    \end{align}
    \end{subequations}
    where $\mathcal X\subset \mathbb R^n,\;\mathcal U \subset \mathbb R^m$ are compact, convex sets with the origin in their interiors. Furthermore, we make the following assumption:
        \begin{ass}
     The pair $(A,B)$ is stabilizable.
     \label{ass:ABStab}
    \end{ass}

    \subsection{Characterization of the steady states and inputs}

We consider the reference command (set-point) tracking problem of bringing the output, state and input of the system to a specific set-point $r\in\mathbb R^p$ and to the associated steady states and inputs $x_{ss},\;u_{ss}$, respectively. Using the usual definition of a steady state and (\ref{eq:IRG_dyn}), the set-points must satisfy the following: 

\begin{equation}
    \begin{bmatrix}
    A-I_n & -B & 0_{n+m\times p} \\
    C & 0_{p\times m} & -I_p 
    \end{bmatrix}  \begin{bmatrix}
    x_{ss}\\u_{ss}\\ r
    \end{bmatrix}=M \begin{bmatrix}
    x_{ss}\\u_{ss}\\ r
    \end{bmatrix}= 0\label{eq:IRG_ss}.
\end{equation}

Assumption \ref{ass:ABStab} ensures that (\ref{eq:IRG_ss}) has a solution \cite{limon2008mpc}. In the following, we define $z_{ss}(r) = (x_{ss}(r),\; u_{ss}(r))$, where $(z_{ss}(r),\; r)$ are solutions to (\ref{eq:IRG_ss}). Given the existence of constraints, the following equation describes an inner approximation of the set of admissible reference commands.
 \begin{equation*}
    { \mathcal R} = \left\{r \in\mathbb R^p\;|\exists \;z\in\tilde{\mathcal Z},\; M\begin{bmatrix}z\\r\end{bmatrix} = 0    \right\}.
 \end{equation*}
where $\tilde {\mathcal Z}\subset \text{ Int } \mathcal Z$ is a compact and convex set. This, under Assumption \ref{ass:ABStab}, implies that ${\mathcal R}$ is compact and convex.

\subsection{Input constrained MPC}
\label{sec:MPCIntro}
As explained in the introduction, uMPC offers several advantages as compared to state and input constrained MPC (cMPC). In the following we consider short-horizon uMPC with a quadratic cost function,

\begin{align*}
    J(\overline\xi,\overline\mu,v) &=\sum_{i = 0}^{N_{\tt MPC}-1} ||\xi_i-x_{ss}(v)||_Q^2 + ||\mu_i- u_{ss}(v)||_R^2 \nonumber\\&\qquad +||\xi_{N_{\tt MPC}}-x_{ss}(v)||_P^2,
\end{align*}
where $\overline \xi = \{\xi_i\}_{i=0}^{N_{\tt MPC}},\;\overline\mu = \{\mu_i\}_{i=0}^{N_{\tt MPC}-1}$, $Q\in\mathbb R^{n\times n},\;R\in\mathbb R^{m\times m},\;P\in\mathbb R^{n\times n}$ and  $N_{\tt  MPC}\in\mathbb N$.  The MPC law is defined using the solution to the following Optimal Control Problem (OCP) $Pr(x,v,N_{\tt MPC})$:
    \begin{subequations}
        \label{eq:IRG_MPC1}
    \begin{alignat}{2}
    \min_{\overline{\xi},\overline{\mu}}&~ &&J(\overline{\xi},\overline{\mu},v) \\
    \text{s.t.}&~ &&~ \xi_0 = x\\
    & &&~\xi_{i+1} = A\xi_i + B\mu_i,\; i = 0,\dots,N_{\tt MPC}-1,\\
    & &&~\mu_i \in \mathcal{U},\quad i = 0,\dots,N_{\tt MPC}-1.
    \end{alignat}
    \end{subequations}
We assume that
    \begin{ass}
     $Q\in\mathbb{S}^n_{++},\;R\in\mathbb{S}^m_{++},\;P \in\mathbb{S}^n_{++}$ and $P=Q+A^\top PA-(A^\top PB)(R+B^\top PB)^{-1}(B^\top PA)$, i.e. $P$ is the solution to the Discrete Algebraic Riccati Equation (DARE). 
    \label{ass:PRicat}
    \end{ass}
    Finally, let 
\begin{equation}
\{u^*_j(x,v,N_{\tt MPC})\}_{j=0}^{N_{\tt MPC}-1} \label{eq:IRG_optSeq} 
\end{equation}
 denote the solution to $Pr(x,v,N_{\tt MPC})$. Then, at time instant $k$ the MPC computed input is given by $ u_{k} = u^*_0(x_{k},v_{k},N_{\tt MPC})$. Assumption 1 and $Q \in \mathbb{S}_{++}^n$ ensure the existence of a stabilizing solution to the DARE in Assumption \ref{ass:PRicat}, and since $0\in\text{Int}\;\mathcal U$ the MPC control law is locally stabilizing at strictly constraint admissible equilibria \cite{borrelli2017predictive}. Note that MPC described in this section does not handle state constraints which will be handled by the IRG. 
 
\subsection{Incremental Reference Governor (IRG)}

For the time being, suppose that a control law for system (\ref{eq:IRG_dyn}), 
\begin{equation} u = g(x,r), \label{eq:IRG_genCtrlLaw}\end{equation}
which depends on the state $x$ and reference command $r$, is available. We define $ u^{g}_j(x,r) = g( x^{g}_j(x,r),r),\; j\in\mathbb N_0$ and $x^g_j(x,r)  = A^j x +\sum_{i=0}^{j-1}A^{j-1-i}Bu_i^{g}$ for $j\geq1$ and $x_0^g(x,r) = x$. The corresponding state-input vector is $ z^{g}_j(x,r) = (x^g_j,u^{g}_j)$.

Now, considering (\ref{eq:IRG_dyn}) in closed-loop with controller (\ref{eq:IRG_genCtrlLaw}), the aim of the IRG is to adjust the reference command that the system follows in such a way as to ensure that constraints are enforced. The IRG accomplishes this by testing whether an increment of the current reference command leads to constraint admissible trajectories.

More specifically, at each time step, the reference increment is parameterized as  $ v^+ = v_{k-1} + \kappa v_{dir}$, where $\kappa\in[0,1]$ is a parameter that dictates the rate at which $v_k$ converges to $r$,
\begin{equation}
    v_{dir} = r - v_{0},\label{eq:IRG_ydir}
\end{equation}
$v_0\in\mathcal R$ is such that $\{z^{g}_j(x_0,v_0)\}_{j=0}^{\infty}$ does not violate constraints and $x_0$ is the initial state.
If the constraints hold for $\{z^{g}_j(x_k,v^+)\}_{j=0}^{\infty}$ then $v_{k} = v^+$, otherwise, $v_{k} = v_{k-1}$.

 For certain problems, e.g.  if the control law (\ref{eq:IRG_genCtrlLaw}) is an LQR and there are only polytopic constraints, it is possible to compute the Maximum Output Admissible Set (MOAS), $\mathcal O_{\infty}^g(v)$, associated with $\;\mathcal Z$, (\ref{eq:IRG_genCtrlLaw}) and $v\in \mathcal R$. The constraint evaluation step is then reduced to verifying
    \begin{equation*}
           x^{g}_0(x_k,v^+) \in \mathcal O^g_{\infty}(v^+).
    \end{equation*}
     However, if $\mathcal O^g_{\infty}(v^+)$ (or a good inner approximation of it) cannot be computed, an alternative approach \cite{bemporad1998reference} is to predict state and control trajectories and verify if  
    \begin{align*}
        &z^{g}_j(x_k,v^+)\in \mathcal{Z}, \;j=0,\dots,N_{\tt RG}-2,\\
        &x^{g}_{N_{\tt RG}-1}(x_k,v^+)\in \mathcal I^g(v^+) ,
    \end{align*}
   where $\mathcal I^g(v^+)\subset \mathcal O^g_\infty(v^+)$ is a forward invariant set that contains $x_{ss}(v^+)$ in its interior and  $N_{\tt RG}\in\mathbb{N}_0$ is a fixed horizon length. Note that, $\mathcal I^g(v^+)$, is potentially small as compared to $\mathcal O^g_\infty(v^+)$. Using the prediction allows to extend the feasible region as entering $\mathcal I^g(v^+)$ is only required after $N_{\tt RG}$ steps.
   
   If the control law (\ref{eq:IRG_genCtrlLaw}) is the uMPC from section \ref{sec:MPCIntro}, computing the MOAS is difficult as the closed-loop system is nonlinear. A prediction-based approach, nevertheless, can be used to implement the IRG. Note, however that at each time instant, to compute the predicted input sequence over $N_{\tt RG}$ steps, one must solve $N_{\tt RG}$ optimization problems of the form \eqref{eq:IRG_MPC1}. This has the potential to be computationally demanding, possibly negating the advantages of using efficient uMPC solvers to alleviate computational burden. In the next section, we introduce the RGMPC scheme that has lower computational requirements.

\section{proposed RGMPC scheme}
\label{sec:nIRG}
Based on the ingredients introduced in the last two sections we now introduce our RGMPC scheme which augments uMPC to handle (potentially non-polyhedral) state constraints and non-polyhedral input constraints, whilst having a low computational effort.

Consider an input sequence, $\{u^{ext}_j(x,v)\}_{j=0}^\infty$, where 
\begin{align}
 u&_j^{ext}(x,v) =\label{eq:IRG_MPCSEC1}\\
 &\begin{cases}
     u^*_j(x,v,N_{\tt MPC})       & \text{if } j < N_{\tt MPC}\\
    \Pi_{\mathcal U}\left[K( x^{ext}_j- x_{ss}(v)) +  u_{ss}(v)\right] &  \text{if } j  \geq N_{\tt MPC} \nonumber
  \end{cases}
\end{align} 
where $K=(B^\top P B + R)^{-1}(B^\top P A)$ is the LQR gain associated with matrices $Q$ and $R$, $P$ is the solution to the associated DARE, $\Pi_{\mathcal U}(\cdot)$ denotes the projection operator onto the set $\mathcal U$, $x^{ext}_j = A^j x + \sum_{i=0}^{j-1}A^{j-1-i}Bu^{ext}_i$ and $v\in\mathcal R$.  Sequence (\ref{eq:IRG_MPCSEC1}) is the optimal input sequence of (\ref{eq:IRG_MPC1})  padded with a saturated LQR law for $j\geq N_{\tt MPC}$. 

Suppose that the sequence (\ref{eq:IRG_MPCSEC1}) has been computed at a time instant $k$ for the reference command $v^+$. A sufficient condition to ensure that this sequence and its associated state trajectory satisfy the constraints is that
\begin{subequations}
\label{eq:IRG_refCond2}
\begin{align}
 &z^{ext}_j(x_k,v^+)\in \mathcal{Z}, \quad j \leq N_{\tt RG}-2 \\
 &x^{ext}_{N_{\tt RG}-1}(x_k,v^+)\in \mathcal{I}^{\tt LQR}(v^+).
 \end{align}
\end{subequations}
where $N_{\tt MPC}$ is typically much smaller than $N_{\tt RG}$,  $\mathcal{I}^{\tt LQR}(v^+)\subset\mathbb R^n$ is a constraint admissible forward invariant set for system (\ref{eq:IRG_dyn}) under the LQR law associated with $Q$ and $R$.
Algorithm \ref{algo:1} describes the proposed RGMPC scheme.

\begin{algorithm}
\caption{Input generation and closed-loop system evolution at time instant $k$. } \label{algo:1}
\begin{algorithmic}[1]
\Require{ $x_k$: the current state, $v_{k-1}$: the  reference used at time $k-1$, $k'$: the last time instant at which $v_{k'}\neq v_{k'-1}$ (default $k' = 0$), $\{u^{ext}_j(x_{k'},v_{k'})\}^{N_{\tt MPC -1}}_{j=0}$, 
and $v_{dir}$. }
\State select $\kappa_k\in[0,1]$
    \State compute $v^+ = v_{k-1} + \kappa_k v_{dir}$ 
    \State compute $\{ u^{ext}_j(x_k,v^+)\}^{N_{\tt RG}-1}_{j=0}$ and $\{ z^{ext}_j\}^{N_{\tt RG}-1}_{j=0}$.
    \If{$\{ z^{ext}_j\}^{N_{\tt RG}-1}_{j=0}$ violates (\ref{eq:IRG_refCond2})} 
      \State $v_k = v_{k-1}$, $\kappa_k = 0$
    \If{$k-k'<N_{\tt MPC}$}
        \State $u_k = u^{ext}_{k-k'}(x_{k'},v_{k'})$
    \Else
    \State $u_k = \Pi_{\mathcal U}\left[K(x_k-x_{ss}(v_k))+u_{ss}(v_k)\right]$
    \EndIf
    \Else
       \State$ v_k = v^+$ 
       \State $u_k = u^{ext}_{0}(x_k,v^+)$ 
       \State $k' = k$
    \EndIf
    \State apply $u_k$ to the system.
    \State \Return {$v_{k}$, $k'$, $\{u^{ext}_j(x_{k'},v_{k'})\}^{N_{\tt MPC -1}}_{j=0}$, $\kappa_k$}
\end{algorithmic}
\end{algorithm}
\begin{rmk}
Algorithm \ref{algo:1} checks constraints for sequence (\ref{eq:IRG_MPCSEC1}) corresponding to the incremented reference command $v^+$. If constraints are satisfied, the incremented reference is accepted, $v_k=v^+$.  If not, the reference is held constant and the corresponding element of the MPC sequence computed at the time instant $k'$ (the last instant the reference command was updated) is applied.  Note that, if RGMPC is not able to update $v_k$ for more than $N_{\tt MPC}-1$ steps, it switches to saturated LQR feedback.
\end{rmk}

\begin{rmk} The choice of the terminal set $\mathcal I^{\tt LQR}(v)$ is application specific. A common choice is the MOAS of the LQR controlled closed-loop system. In the case of polytopic constraints the MOAS is also polytopic and can be computed in closed form \cite{gilbert1991linear}. For non polytopic constraints, if a polytopic approximation is possible, the problem is reduced to the previous case. Another choice for $\mathcal I^{\tt LQR}(v)$ are constraint admissible sublevel sets of Lyapunov functions of the LQR controlled system. If $P$ is the solution to the Lyapunov equation: $(A+BK)^\top P(A+BK) - P = I$ then sets of the form $
\mathcal I^c(v) = \{x\;|\;||x-x_{ss}(v)||^2_ P \leq c)\}$ are forward invariant. We can then choose $c$ s.t. $\mathcal I^c(v)\subseteq \mathcal Z$. This is a specific case of sets used in the RGs introduced in \cite{gilbert2002nonlinear}.
\end{rmk}

The values of $\kappa_k$ in line 1 of Algorithm \ref{algo:1} must be carefully selected. For example, if some constraints are active in specific regions of the state space, entering that region may require a smaller reference increment. Conversely, to accelerate the response, we usually look for the largest $\kappa_k$ that is admissible. In reference governors, the choice of $\kappa_k$ is often resolved by solving an optimization problem: maximize $\kappa_k$ such that the corresponding reference increment leads to a closed-loop state and input sequence that is constraint admissible \cite{garone2017reference}. In Algorithm \ref{algo:2}, we propose a simple $\kappa_k$ selection logic to ensure that a reference increment is feasible in finite-time without the need for the RG optimization problem to be solved. 
\begin{algorithm}
\caption{Selection of reference increment, $\kappa_k$, for Algorithm \ref{algo:1}} \label{algo:2}
\begin{algorithmic}[1]
\Require{ $\kappa^0\in(0,1]$: a default value of the increment. $N_{\tt a}\in\mathbb N$ a tuning parameter, $k$: the current time step, $k'$: the last time step s.t. $v_{k'}\neq v_{k'-1}$, $v_{k-1}$, $v_0$, $r$ and $\{\kappa_j\}_{j=0}^{k}$}
    \If{$k-k'\leq N_{\tt a}$}
    \State $\kappa_k  = \kappa^0$
    \Else
    \State $\kappa_k = \frac{\kappa^0}{k-k'-N_{\tt a}}$
    \EndIf
    \If{$\sum_{j=0}^{k}\kappa_j>1$}
    \State $\kappa_k = \kappa_r$, where $\kappa_r = 1- \sum_{j=0}^{k-1}\kappa_j$, so that $r = v_{k-1}+ \kappa_r( r-v_0)$
    \EndIf
    \State \Return{$\kappa_k$: to be used in Algorithm \ref{algo:1} at time step $k$}
    \end{algorithmic}
\end{algorithm}

For $r\in\mathcal R$ we define the set 
\begin{align*}
    \mathcal{P}(r) = &\left\{x\in\mathbb{R}^{n}\;| \;  z^{ext}_j(x,r)\in \mathcal{Z} \;\forall j \geq 0 \right\}\cap\nonumber\\
    &\left\{x\in\mathbb{R}^{n}\;| \;  x^{ext}_{N_{\tt RG}-1}(x,r)\in \mathcal{I}^{\tt LQR}(r) \right\},
\end{align*}
as the set of states for which the sequence generated by control (\ref{eq:IRG_MPCSEC1}) satisfies (\ref{eq:IRG_refCond2}). We assume that:\\
\begin{ass}
    \label{ass:ballEps} $\exists \epsilon>0$ s.t. $\forall v\in {\mathcal R}$, $\mathcal{B}( x_{ss}(v),\epsilon)\subseteq\mathcal{P}(v)$. 
\end{ass}
We also introduce the set 
\begin{equation*}
    \Gamma = \{(x,v)\in \mathbb R^n\times\mathcal R\;|\; x\in\mathcal P(v)\},
\end{equation*}
of state and reference couples for which (\ref{eq:IRG_refCond2}) is verified.
We now study some theoretical properties of RGMPC as defined by Algorithm \ref{algo:1} and Algorithm \ref{algo:2}. To facilitate this analysis, we first establish some preliminary results.
\begin{lem}
\label{lem:EnterEpsiInFiniteTime_ASsys}
Given an asymptotically stable (A.S.) linear system $x_{k+1} = A_c x_{k}$, $x\in\mathbb{R}^n$, and a compact set $\mathcal S\subset \mathbb R^n $ with the origin in its interior, it follows that $\forall \delta>0,\;\exists N\in\mathbb N$ s.t. $\forall j\geq N$, $A_c^{j}x_0\in\mathcal B(0,\delta),\; \forall x_0\in\mathcal S$.
\end{lem}
\begin{proof}
Given the system is A.S., following classical Lyapunov stability results for discrete linear systems, $A_c^\top PA_c - P = - I$ has a unique solution $P\in\mathbb S^n_{++}$. 
Define the Lyapunov function $V(x) = \frac{1}{2}x^\top P x$. By property of A.S. for linear systems, $\exists q\in(0,1)$ s.t. $V(A_c^i x)\leq q^i V(x),\;\forall i>1$. Now, let $c_1 = \max\{ V(x)\;|\;x\in\mathcal S\}$, which exists given continuity of $V$ and compactness of $\mathcal S$, and let $c_2>0 $ s.t. $\{x\;|\;V(x) \leq c_2\}\subseteq\mathcal B(0,\delta)$. Then, choosing $N = \min\{ j\;|\;c_2\geq q^j c_1\}$ completes the proof. 
\end{proof}
\begin{lem}
\label{lem:EnterEpsiInFiniteTime_1Ref}
Given Assumptions \ref{ass:ABStab} and \ref{ass:PRicat}, and $v\in\mathcal R$ it follows that $\forall \delta>0,\;\exists  N_{\delta,v} \;\in\mathbb N$ s.t. $\forall j\geq N_{\delta,v},\;x^{ext}_j(x,v)\in\mathcal B(x_{ss}(v),\delta),\;\forall \;x\in\mathcal P(v)$.
\end{lem}
\begin{proof} We define $x_j^{ext} = x_j^{ext}(x,v)$. Given that (\ref{eq:IRG_refCond2}) holds for $\{x^{ext}_j\}^{N_{\tt RG-1}}_{j=0}$, then $x^{ext}_{N_{\tt RG}-1}\in\mathcal I^{LQR}(v)\subseteq\mathcal O_{\infty}^{LQR}(v)$. Hence, for all $j\geq N_{\tt RG}-1$, $u^{ext}_{j}$ is derived from an LQR with gain matrix $K$. This, combined with Assumption \ref{ass:ABStab}, makes  $\{x^{ext}_j\}_{j=N_{\tt RG}-1}^\infty$ equivalent to a trajectory of (\ref{eq:IRG_dyn1})controlled using LQR with  A.S. equilibrium $x_{ss}(v)$. Making the change of variable $\tilde x = x - x_{ss}(v)$, the dynamics of the associated LQR controlled system are given by $\tilde x_{k+1} = (A-BK)\tilde x_k$. Also, note that $\mathcal O_{\infty}^{LQR}(v)$ is compact given compactness of $\mathcal X, \mathcal U$, by \cite[Theorem 2.1 (i)]{gilbert1991linear}. Therefore, for all $v\in\mathcal R$, Lemma \ref{lem:EnterEpsiInFiniteTime_ASsys} states that for the system $\tilde x_{k+1} = (A-BK) \tilde x_k$ and associated $\mathcal O_{\infty}^{LQR}(v)$, $\forall\delta>0$ $\exists N$ s.t. $\forall j\geq N,\; (A-BK)^j x\in\mathcal B(0,\delta),\;\forall x\in\mathcal O^{LQR}_{\infty}(v)$.
Introducing $N_{\delta,v} = N_{\tt RG}+N$ directly implies $\forall j\geq N_{\delta,v},\;x^{ext}_{j}(x,v)\in\mathcal B(x_{ss},\delta), \;\forall x\in\mathcal P(v)$.
\end{proof}
\begin{lem}
\label{lem:EnterEpsiInFiniteTime}
Given Assumption \ref{ass:ABStab} and \ref{ass:PRicat}, it follows that  $\forall\delta>0,\;\exists  N_\delta \;\in\mathbb N$ s.t. $\forall j\geq N_\delta$, $\;x^{ext}_j(x,v)\in\mathcal B(x_{ss}(v),\delta),\;\forall (x,v)\in\Gamma$.
\end{lem}
\begin{proof}
From Lemma \ref{lem:EnterEpsiInFiniteTime_1Ref} $\forall v \in\mathcal R\;
    \exists N_{\delta,v} \in\mathbb N$ s.t. $\forall j\geq N_{\delta,v},\;x^{ext}_{j}(x,v)\in\mathcal B(x_{ss}(v),\delta),\;\forall x\in \mathcal P(v)$.
As stated in the proof of Lemma \ref{lem:EnterEpsiInFiniteTime_1Ref}, the rate of decay of the A.S. system associated with $\{x^{ext}_{j}\}_{j=0}^\infty$ does not depend on the reference. Instead, $N_{\delta,v}$ depends on $v$ through the size of $\mathcal O^{LQR}_\infty(v)$. In other terms, the set $\mathcal S$ and associated $c_1$ in Lemma \ref{lem:EnterEpsiInFiniteTime_ASsys} change with $v$. However, $\forall v\in\mathcal R,\;\mathcal O^{LQR}_\infty(v)\subseteq\mathcal X$. Thus, given compactness of $\mathcal X$ and continuity of the weighted distance between two points: $\mathbb R^n\times\mathbb R^n\to \mathbb R_+,\;d_P(a,b) = ||a-b||_P$, then $d_P(\cdot,\cdot)$ reaches a maximum over $\mathcal X\times\mathcal X$. We can bound the $c_1$ constants by $c_M = \frac{1}{2}\max\{d_P(a,b)\;|\; a,b,\in\mathcal X\times\mathcal X\}$. In a similar way to the proof of Lemma \ref{lem:EnterEpsiInFiniteTime_ASsys}, we define $N_\delta=  \min\{ j\;|\;c_2\geq q^{j}c_M\}$. Therefore $N_{\delta,v}\leq N_\delta,\;\forall v\in\mathcal R$.
\end{proof}

Before presenting the next Lemma we introduce  $\mathcal R_{\bar\kappa} = \{v\in\mathcal R\;|\; \{v\}+\mathcal B(0,\bar\kappa)\subseteq\mathcal R\}$ where $\bar\kappa>0$.
\begin{lem}
\label{lem:FTConvR}
 If Assumptions \ref{ass:ABStab}-\ref{ass:ballEps} hold, there exists $ N_{\epsilon/2}\in\mathbb N$ and $\;\bar\kappa>0$, s.t. given $v\in\mathcal R_{\bar \kappa}$ and $x\in\mathcal P(v)$ then $\forall\; \kappa \in[0,\bar\kappa],\;\forall j\geq N_{\epsilon/2}$, $x^{ext}_{j}(x,v)\in \mathcal P(v^+)$, where $v^+ = v + \kappa v_{dir}$, $v_{dir} = r-v_0$ and $\{v_0,\;r\}\subset\mathcal R_{\bar\kappa}$. Moreover, the constant $\bar \kappa$ is independent from $x$.
\end{lem}
\begin{proof}
Select $\epsilon$ from Assumption \ref{ass:ballEps}. By Lemma \ref{lem:EnterEpsiInFiniteTime} $\exists N_{\epsilon/2}$ s.t.
\begin{align}
    \forall j\geq N_{\epsilon/2},\; x^{ext}_{j}(x,v) \in \mathcal{B}(x_{ss}(v),\epsilon/2),\;\forall x\in\mathcal P(v).\label{eq:IRG_Lem2L1}
\end{align}
    Now, let $x^{ext}_{j}=x^{ext}_{j}(x,v)$. Then, for any $j\geq N_{\epsilon/2}$ consider  
    \begin{align*}
    ||x^{ext}_{j}-& x_{ss}(v^+)|| =  ||x^{ext}_{j}- x_{ss}(v) - \kappa x_{ss}(v_{dir})||\\
   \leq& ||x^{ext}_{j}- x_{ss}(v)|| + || \kappa  x_{ss}(v_{dir})|| ,\\
    \leq& \frac{\epsilon}{2} + \kappa||x_{ss}(v_{dir})||.
    \end{align*}
    Where the second line follows from the triangle inequality and the third from (\ref{eq:IRG_Lem2L1}). Defining: $\bar \kappa = \frac{\epsilon}{2||x_{ss}(v_{dir})||}$
    \begin{equation*}
    ||x^{ext}_{j}- x_{ss}(v^+)|| \leq  \epsilon,\;\forall \kappa\in[0,\bar\kappa].
    \end{equation*}
    Therefore, $x^{ext}_{j}(x,v)\in\mathcal P(v^+)$, by Assumption \ref{ass:ballEps}.
    \end{proof}
\begin{rmk}
Note that the maximum step size, $\bar\kappa||x_{ss}(v_{dir})||$ depends only on $\epsilon$. The set $\mathcal R_{\bar\kappa}$ can be made arbitrarily close to $\mathcal R$ by decreasing the value of $\bar\kappa$. This is achieved by decreasing the value of $\epsilon$. Validity of Assumption \ref{ass:ballEps} is still ensured. Also, compactness of $\mathcal R$ is inherited by $\mathcal R_{\bar\kappa}$ \cite[Theorem 2.1 (x)]{kolmanovsky1998theory}. Also,  in Lemma \ref{lem:FTConvR} the convex hull of $\{v_0,r\}$ lies inside $\mathcal R_{\bar\kappa}$. This is relevant, as in Algorithm \ref{algo:1}, $v^+$ is inside the convex hull of $\{v_0,r\}$.
\end{rmk}    
\begin{lem}
\label{lem:fRec} Consider system (\ref{eq:IRG_dyn}) with IC $x_0$ and  desired set-point $r \in\mathcal R$. If $\exists v^0\in\mathcal R$ s.t. $(x_0,v^0)\in\Gamma$, initializing $v_0 = v^0$ in Algorithm \ref{algo:1} and defining $v_{dir}$ according to (\ref{eq:IRG_ydir}) ensures recursive feasibility of Algorithm \ref{algo:1}: $x_k\in\mathcal X\Rightarrow z_{k+1}\in\mathcal Z,\;\forall k\geq 0$.
\end{lem}
\begin{proof}
The claim follows directly from the assumptions and implementation of Algorithm \ref{algo:1}. Using Algorithm \ref{algo:1}, the trajectory of the system between two subsequent reference increments at times $k_1,\;k_2\in\mathbb N$, $k_1<k_2$ is given by $\{z^{ext}_j(x_{k_1},v_{k_1})\}_{j=0}^{k_2-k_1}$. For a reference increment to be performed, at time $k_2$, $\{z^{ext}_j(x_{k_2},v_{k_2})\}^\infty_{j=0}\subseteq\mathcal Z$ is required. Finally, $(x_0,v_0)\in\Gamma$ implies that $\{z^{ext}_j(x_0,v_0)\}^\infty_{j=0}\subseteq\mathcal Z$. Therefore $z_k\in\mathcal Z,\;\forall k\geq0$. As a result recursive feasibility is ensured.
\end{proof}

\begin{thm}
\label{thm:mainTheo}
Consider the problem of bringing system (\ref{eq:IRG_dyn}), controlled using Algorithm \ref{algo:1}-\ref{algo:2}, to the final set-point $r\in\mathcal R_{\bar\kappa}$ which is constant in time, from the initial state $x_0$ subject to constraints (\ref{eq:IRG_cstrSet}). Assume that Assumptions \ref{ass:ABStab}-\ref{ass:ballEps} hold,  and that $\exists v^0$ s.t. $(x_0,v^0)\in\Gamma$. If $v_0 = v^0$ and $v_{dir}$ is defined according to (\ref{eq:IRG_ydir}), then finite-time convergence of $v_k$ to $r$ and asymptotic convergence of the state, $x_k$, to $x_{ss}(r)$ is ensured.
\end{thm}
\begin{proof} Suppose $v_0\neq r$. We then need to show there exists $ k^*\in\mathbb N$ s.t. $\forall j\geq k^*,\;v_{j} = r$. Define
\begin{equation*}
     \Delta v_k =v_k  - v_0 = s_k v_{dir},\quad s_k=\sum_{j=0}^k\kappa_j,
\end{equation*}
Since $\kappa_k\geq0\;\forall k$, showing finite-time convergence to $r$ is equivalent to showing  that $\exists k^*\in\mathbb N$ s.t. $\forall j\geq k^*$, $ s_j = 1$. We do this by contradiction.\\
Hypothesis (\textbf{H}): $\nexists k^*\in\mathbb N$ s.t. $\forall j\geq k^*,\;s_j = 1$.
First, define $N_{\epsilon/2}$ as in Lemma \ref{lem:FTConvR}, $N_{\bar\kappa} = \min\{ i\in\mathbb N\;|\;\frac{\kappa^0}{i-N_a}\leq \bar\kappa\}$ and define $N = \max\{N_{\epsilon/2},\;N_{\bar\kappa}\}$. At any time instant, $k_1\in\mathbb N$ consider the last instant such that there was a change in the reference: $k'= \max \{i\leq k_1\;|\; v_{i-1}\neq v_{i}\}$. Now, assume that $v_{k'+N-1}=v_{k'}$, then, following Algorithm \ref{algo:1}-\ref{algo:2}, $x_{k'+N} = x^{ext}_{N}(x_{k'},v_{k'})\in\mathcal P(v^+_{k'+N})$, where $v^+_{k'+N}$ is the tested reference at time $k'+N$. This is because, from Lemma \ref{lem:FTConvR}, $\forall j\geq N_{\epsilon/2}$, $x^{ext}_{j}(x_{k'},v_{k'})\in\mathcal P(v_{k'}+\bar\kappa v_{dir})$,  and because at time $N$ the tested increment is smaller than $\bar\kappa$. This implies $v_{k'+N}\neq v_{k'}$ and $\kappa_{k'+N}\geq\frac{\kappa_0}{N}$. Now, if an advance of the reference takes place at any $j\in\mathbb N\cap[k',k'+N]$, then $\kappa_k\geq\frac{\kappa_0}{j-k'-N_a}\geq\frac{\kappa^0}{N}$ by the implementation of Algorithm \ref{algo:2}. Hence, the reference is incremented at least every $N$ steps. Thus, there exists an infinite sequence of time instants $\{k_i\}_{i=0}^\infty$ s.t. $\kappa_{k_i} \geq\frac{\kappa^0}{N}$. In turn, the sequence of $\{s_{k_i}\}^{\infty}_{i=0}$ diverges. Now, choose $k^*=\min k$ s.t. $s_k \geq 1$. At that time instant, Line 6 of Algorithm \ref{algo:2} is executed, and given that $\kappa_r\leq\bar\kappa$ an increment of $\kappa_r$ is performed, ensuring $s_{k^*} = 1$. Lines 5-6 of Algorithm \ref{algo:2} ensure that $\forall j\geq k^*$ $s_j=1$, violating \textbf{H}. As such we have $k^*>0$ s.t. $\forall j\geq k^*s_j=1$. Thus, convergence of $v_k$ to $r$ in finite-time is proved.\\
Convergence of the state to $x_{ss}(r)$ is directly implied using Lemma \ref{lem:EnterEpsiInFiniteTime}.
\end{proof}
\begin{rmk}
\label{rem:kappaSelec}By examining the proof of Theorem \ref{thm:mainTheo}, it can be shown that the results hold for other $\kappa_k$ selection strategies as long as such strategies  ensure that whenever $k-k^\prime$ becomes large, $\kappa_k \leq \bar{\kappa}$.
\end{rmk}

\section{Illustrative example}
\label{sec:EXAMPLE}
We consider a problem of spacecraft rendezvous to a target on a circular orbit. The relative motion dynamics are represented by the CWH equations \cite{clohessy1960terminal} given by
\begin{subequations}
\label{eq:ex_genRelMo}
\begin{align}
    \dot x &= A_c x + B_c u,\\
    y &= C_c x,
\end{align}
\end{subequations}
with $x\in \mathbb R^6$, $u\in \mathbb R^3$, describing the relative motion of the spacecraft in the Hill's frame centered at the target. The first three and last three states represent radial, along track and cross track positions and velocities of the spacecraft, respectively. The inputs are relative accelerations (normalized thrust: $[{\tt N}\;{\tt kg}^{-1}]$) along the three axes. In (\ref{eq:ex_genRelMo}),
\begin{align*}
     A_c &= \begin{bmatrix}
    \; & 0_{3\times3} &\;&\;&  I_{3\times 3}\\
    \\
    3n^2 & 0 &0 &0 & 2n &0\\
     0 &0 &0 & -2n &0 &0\\
     0 &0 &-n^2 & 0 &0 &0 \\
\end{bmatrix},\\
B_c &= \begin{bmatrix}  0_{3\times3}\\ I_{3 \times 3}\end{bmatrix} ,\quad C_c = \begin{bmatrix} I_{3\times 3}& 0_{3\times3} \end{bmatrix},
\end{align*}
where $n =\sqrt{\mu/r_0^3} $, $\mu$ is the gravitational parameter and $r_0$ is the orbital radius of the nominal orbit.

The system has the following state and control constraints:
\begin{itemize}
    \item Input saturation: $||u||_\infty \leq 0.1$. 
    \item Maximum speed: $|x_i|\leq 3,\;i= 4,\;5,\;6$.
    \item The spacecraft must remain in front of the target in the in-track direction, $x_2 \geq 0 $.
    \item Line of sight cone (nonlinear, convex constraint): The spacecraft should remain in the 15 deg cone defined by $x_1^2 + x_3^2 - \tan^2 15^o (x_2+1)^2\leq 0$.
    \item Final speed (\textit{if-then} constraint): When approaching the target, the norm of the relative velocity should be small enough to avoid damage: If $x_2\leq 2$ then $x_4^2+x_5^2+x_6^2 \leq 0.1^2$
\end{itemize}

The spacecraft relative motion dynamics have forced equilibria of the form:
\begin{align*}
    r &= [a,\;b,\;c]^\top, \quad a,\;b,\;c\in \mathbb R\\
    \hat u_{ss}(r) &= [-3n^2a,\;0,\;n^2c] ^\top,
\end{align*}
where, the elements of $r$ correspond to the output states in $y$ and the relative speed is zero at all forced equilibria.
Simulations are performed considering a nominal orbit at 500 km altitude above the earth. When discretizing the linearized system, a sampling period $T_s =0.5$ sec is used. 
When we apply Algorithm \ref{algo:1}, the MPC has a prediction horizon of $N_{\tt MPC} = 20$ steps and the constraint satisfaction is assessed over a horizon of $N_{\tt RG} = 120$ steps. To choose the value of $N_{\tt RG}$, a set of 200 randomly generated ICs with initial velocity norm lower or equal to $1.5\;[\tt m \;\tt s ^{-1}]$, $x_2 \in \;[50,150]\;[\tt m]$ and $x_1,\;x_3$ inside the cone were generated. The value of $N_{\tt RG}$ was chosen such that for all the ICs that had a constraint admissible initial reference, no constraint violation occurred during subsequent 150 sec of simulation.

The MPC weight matrices were chosen as $R =I_3$, $Q = \text{diag}([100,\;1,\;100,\;10,\;1,\;10])$.
By relying on flexibility in choosing $\kappa_k$ described in Remark 4, we utilized the following scheme for computing $v^+$ which is better adopted to the problem at hand: 
\begin{subequations}
\begin{align}
&v^+ = v_{k-1} + \kappa \Delta v, \label{eq:CWH_IRGRefPlus}\\
& \Delta v   = \begin{cases}
    \text{sign}(r-v_0)\cdot\Delta v_{fix}    & v_{-,2}\geq 20  \;[\tt m], \\
    r -v_{-}  & v_{-,2} < 20  \;[\tt m], \label{eq:CWH_IRGRefPlus}
  \end{cases}\\
    &\Delta v_{fix} = r - [3.67,\;20,\;3.67]^\top,\;\kappa = 0.1,\label{eq:CWH_IRGRefPlusFinal}
\end{align}
\end{subequations}
where $r = 0_{3\times 1} $ is the final set-point, and $v_{k-1}$ denotes the previous set-point. The reference is incremented by a fixed amount when far from the target ($>20\tt \;[m]$) and proportionally to the difference between $v_{k-1}$ and $r$ when close to the target. Initialization of the reference is done by setting $v_0 = Cx_0$ where $x_0$ is the IC.

Figure \ref{fig:simple_IRGMPC_traj} shows the time histories of states, reference commands and inputs for the spacecraft starting at $x(0) = [10,\;100,\;20,\;0,\;0,\;0]^\top$ and controlled by the proposed RGMPC scheme with the Fast MPC solver \cite{kogel2011fast}. The simulation shows convergence to the target spacecraft in 100 sec while respecting constraints on both states and inputs. Figure \ref{fig:simple_IRGMPC_trajJ} (bottom, left and center) shows two dimensional projections of the trajectory as well as of the line of sight cone constraints which are respected at all times. Finally, Figure \ref{fig:simple_IRGMPC_trajJ} (bottom right) depicts the velocity norm for times from around 75 sec and onward as well as the terminal velocity constraint when it is active. The velocity norm rides the constraint boundary before going to 0 as the spacecraft converges to the final set-point.

Figure \ref{fig:simple_IRGMPC_trajJ} (top) shows the instants at which the reference is changed during the maneuver. After 95 sec the final set-point is reached. In most cases when the reference is held constant, it remains only for 1 or 2 time instants. Only in 2 occasions does this occur for a significantly longer period: for 8 time steps (approx. 10 sec after the start) and for 15 time steps (approx. 60 sec after the start). Hence, with the proposed reference switching logic, the saturated LQR is not used.

\begin{figure}[ht]
\centering
\includegraphics[scale=0.6]{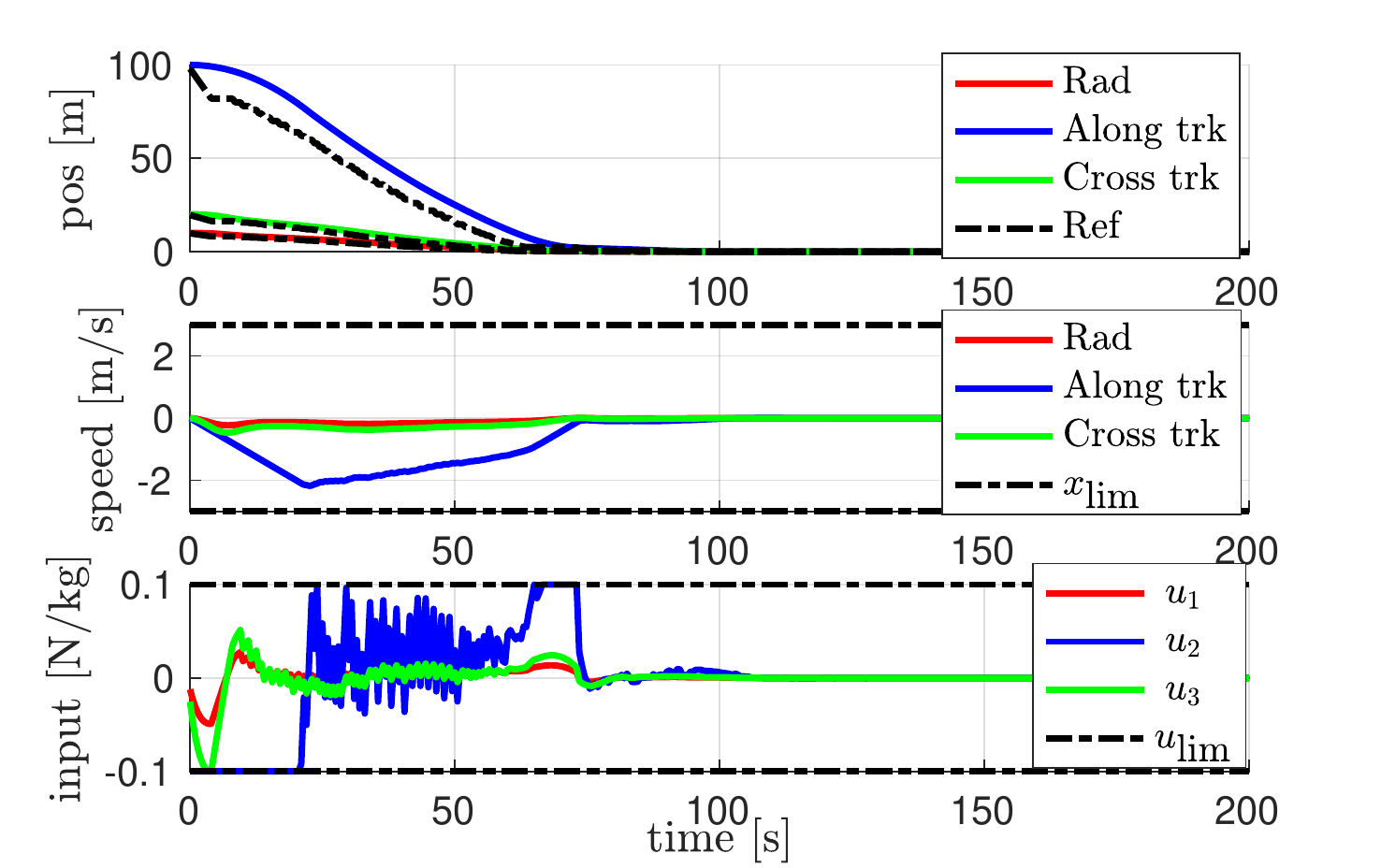}
\caption{Time histories of the state, input and reference signals for the spacecraft. Max-Min constraints on velocities and inputs  are also shown on the lower two figures (dotted lines).}
\label{fig:simple_IRGMPC_traj}
\end{figure}

\begin{figure}[ht]
\centering
\includegraphics[scale=.4]{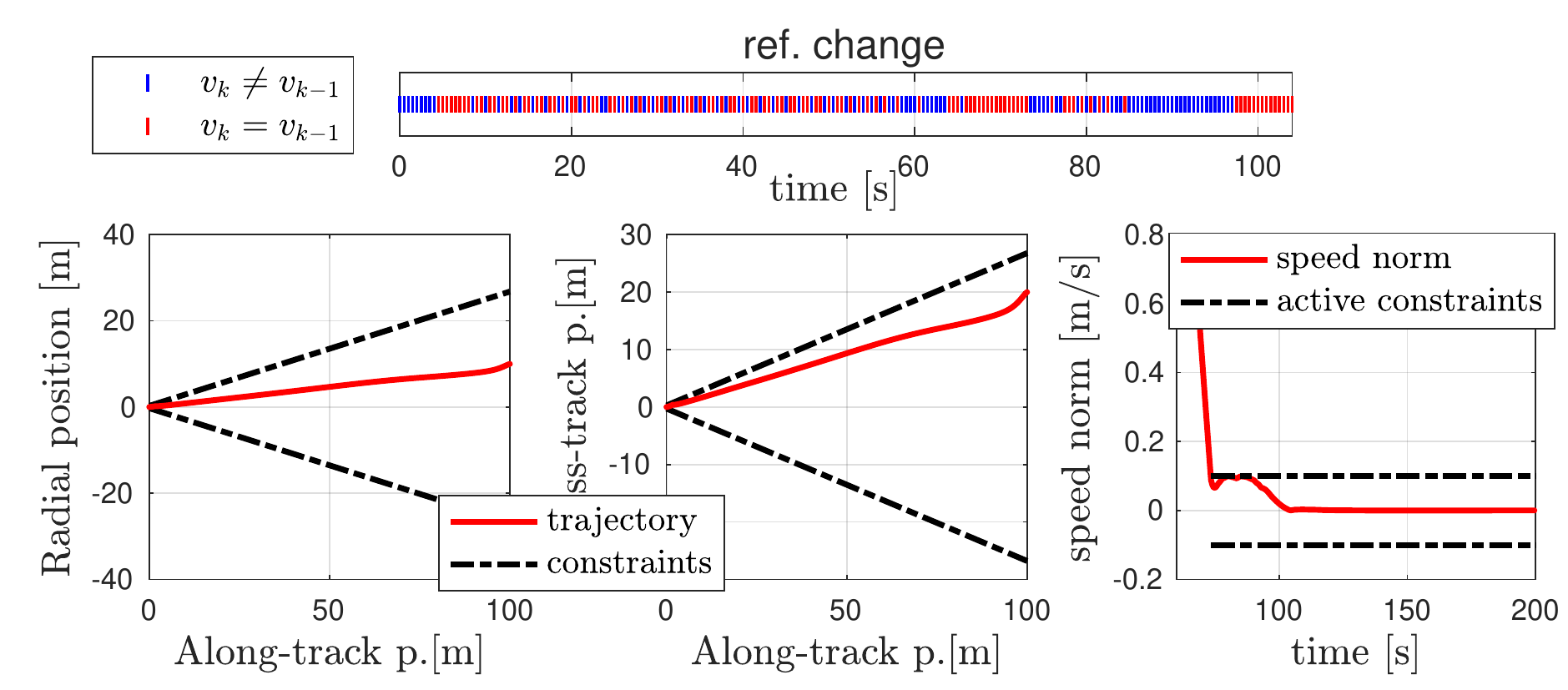}
\caption{Top: Reference changes over the first $100$ sec. Each time instant is represented by a bar the color of which indicates if there has been a reference change. 
Bottom: 2D projections of spacecraft trajectory and speed norm with related constraints.
}
\label{fig:simple_IRGMPC_trajJ}
\end{figure}

\subsection{Comparison to the Fast-MPC without add on scheme}
\label{sec:res_compbasic}
To confirm the necessity of a state constraint handling mechanism we perform simulations over a
grid of IC with either an uMPC or the proposed RGMPC. We consider, at a distance $x_2 = 50 \;[\text m]$, 200 points forming concentric circles in the $x_1$-$x_3$ plane. The radii go up to $r^2 = {\tan^2(14.5^o)(50^2+1)}$. This set of values combined with zero initial speed is used as the set of ICs.
Simulations resulted in the RGMPC not violating constraints a single time while the uMPC violated constraints for each IC. In particular, for each IC the spacecraft passed behind the target spacecraft and the terminal speed constraint  was violated. Figure \ref{fig:MPCComp_cone} depicts what ICs lead to violation of the cone constraint by the uMPC controller. As expected this is often when starting away from the center line of the cone.

\begin{figure}[ht]
\centering
\includegraphics[scale=0.4]{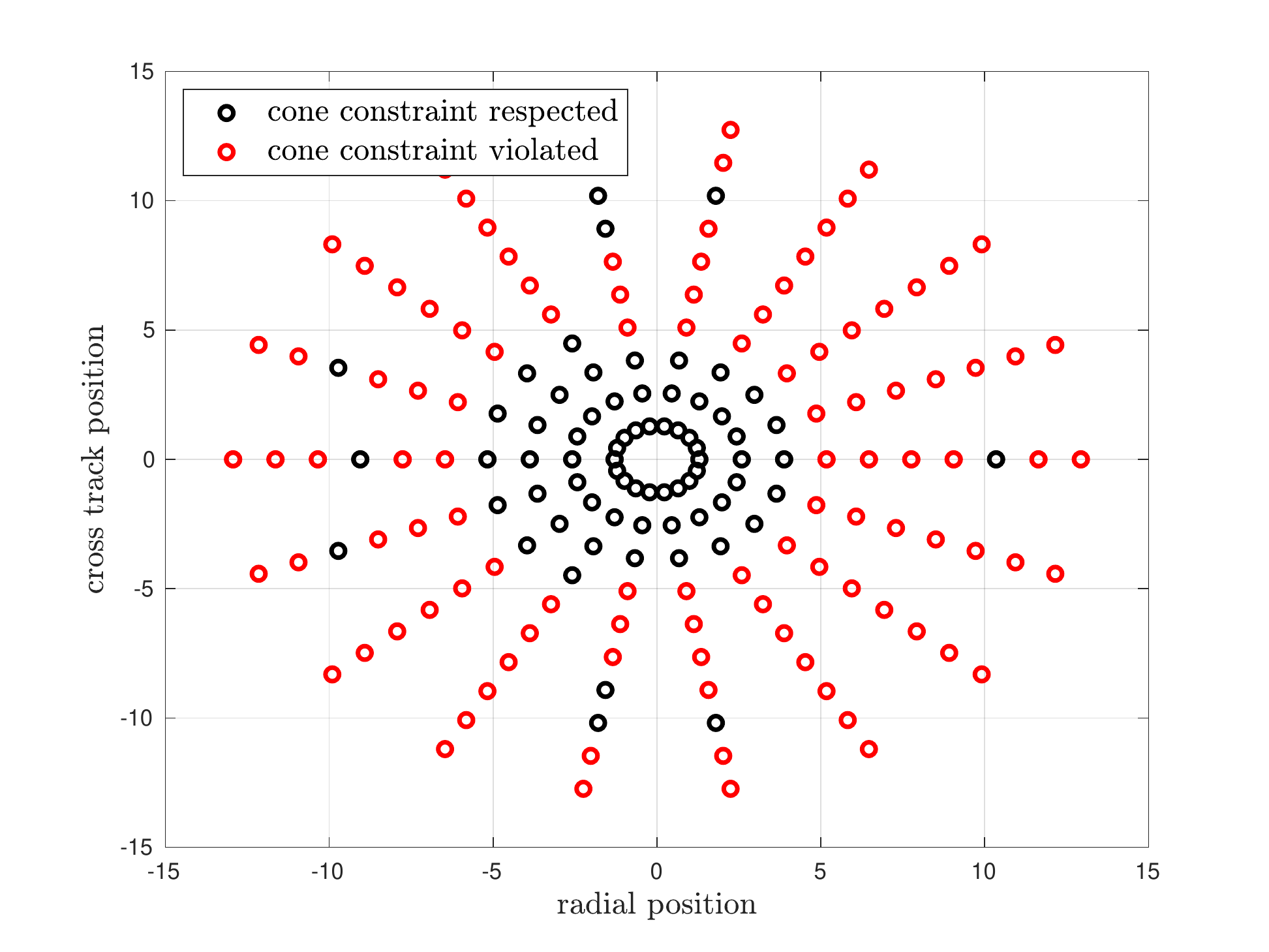}
\caption{ Values of the IC that lead to subsequent violation of the cone constraint when considering uMPC. All ICs have $x_2 = 50\;[\text{m}]$ and zero speed. For all points shown, the RGMPC satisfies the constraints at all times.}
\label{fig:MPCComp_cone}
\end{figure}

\subsection{Comparison to a saturated LQR-IRG scheme}
To assess the advantages of the presented scheme with respect to more conventional schemes we will compare the performance of the RGMPC scheme with that of a saturated LQR extended with an IRG, referred to as sLQR-RG. To do so, we consider the following metrics: 
\begin{itemize}
    \item Successful initialization of the RG and no constraint violation, denoted as ``succ. sim'' type Boolean.
    \item Time required to reach the target spacecraft within a specified tolerance, denoted as $t_{conv}\;[\text s]$.
    \item An input cost that relates to fuel consumption \cite[Section~14.3]{gurfil2016celestial}, computed as $u_{cost} = \int_0^\infty ||u(t)||_2^2\d t\;[\tt N^2\;\tt{kg}^{-2}\;\tt{s}]$.
\end{itemize}
Taking the same uniform grid of ICs as in the previous section and for the same $Q$ and $R$ matrices we performed simulations with both the RGMPC and sLQR-RG controllers. Table \ref{tab:LQRComp_at50m} summarizes results for the different metrics for the two controllers. The differences in input cost and time of convergence are also shown. Succesful simulations are achieved for every IC and both controllers. It is notable that the RGMPC outperforms the sLQR-RG in every single simulation both in maneuver time and in fuel cost. In particular, the RGMPC provides a mean reduction of 21\% in maneuver time and of 70\% in fuel consumption, both substantial values.

\begin{table}[ht]
  \centering
  \begin{tabular}{l| *{3}{c}}
     & \# succ. sim. & mean $u_{cost}$&mean $t_{conv}$\\ \hline
     RGMPC& 200 & 0.9& 75.91 \\
     sLQR-RG& 200 & 2.98& 95.74\\ 
  \end{tabular}
  \caption{Number of successful simulations and mean values of $u_{cost}$ and $t_{cost}$ for the RGMPC and sLQR-RG schemes.}
  \label{tab:LQRComp_at50m}
\end{table}

To explain the difference in performance between the sLQR-RG and RGMPC we consider the state and input trajectories. Figure \ref{fig:MPCLQR_at50} shows the radial component of the state (top) and input (bottom) trajectories for a single IC for the sLQR-RG (red) and RGMPC (blue). One can observe that unlike the RGMPC, the sLQR-RG generated input is prone to oscillations between the saturation values. This oscillation is directly translated into the position evolution as depicted in Figure \ref{fig:MPCLQR_at50}.

\begin{figure}[h]
    \centering
    \includegraphics[scale=0.6]{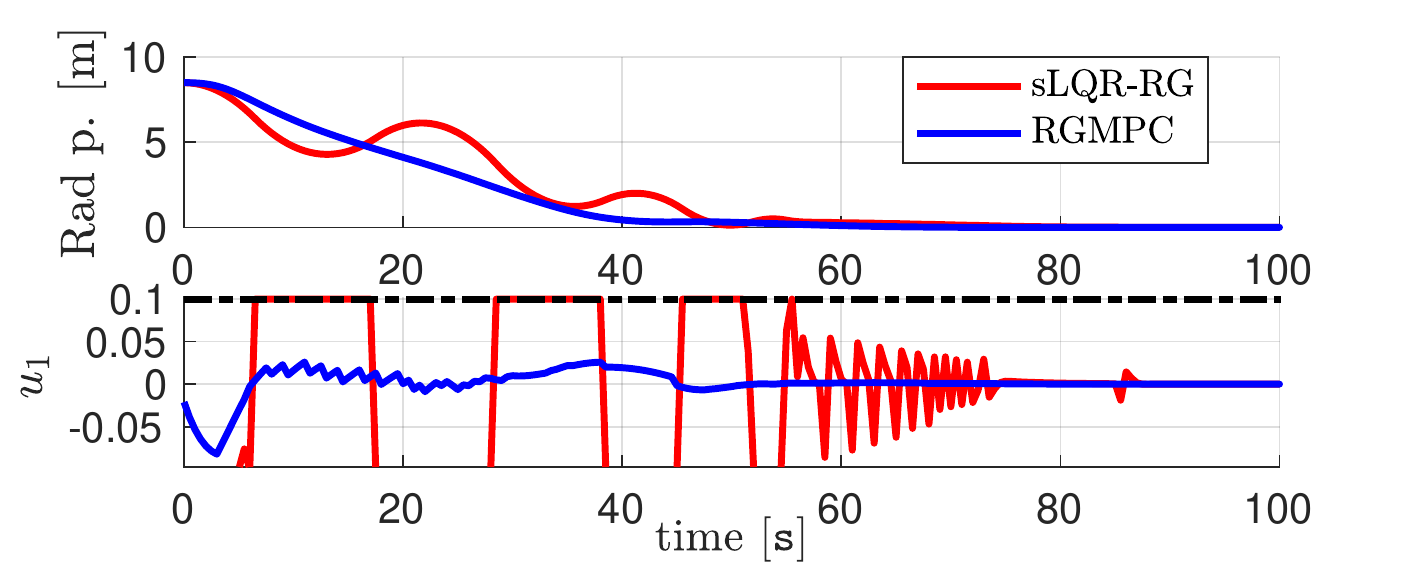}
  \caption{Radial component of position (top) and input thrust (bottom) evolution when considering the RGMPC (blue) and sLQR-RG (red)}
  \label{fig:MPCLQR_at50}
\end{figure}

\subsection{Comparison to state and input constrained MPC}
\label{sec:res_comp_cMPC}
We next assess the viability of the proposed RGMPC scheme in comparison to a state and input constrained MPC, referred to as cMPC. In this section, all OCPs, both for the RGMPC and the cMPC, are solved with a dual active set solver \cite{goldfarb1983numerically}. Once again, we compare the set of trajectories obtained starting at the 200 initial conditions described in Section \ref{sec:res_compbasic}. Additionally, we also use a second, similar set of trajectories starting from $x_2= 100\;[\text m]$. To estimate computational power requirement we collect, at each time step, the time required to compute the control input: $t_{comp}$. For one simulation, the average time required to compute the input commands is referred to as $t_{comp,av}$. Simulations were performed using Matlab on a machine with a 2.3 GHz 8-Core Intel Core i9. 

The RGMPC has horizons $N_{\tt MPC} =20$, and $N_{\tt RG}= 120$, three cMPC formulations were used for comparison. The cMPC differed in their horizon lengths: $N_1 = 20$, $N_2=60$ and $N_3=120$. Additionally, constraints were made polyhedral by the following modifications:
\begin{itemize}
    \item A polyhedral approximation of the line of sight cone: using 15 linear inequalities.
    \item The \textit{if-then} terminal constraint on speed is avoided by setting the  terminal reference to $r = [0\;4\;0]^\top$ and imposing $x_2\geq 3$. 
\end{itemize}
The rest of the constraints as well as the rest of the simulation parameters were kept identical to those of previous sections.

For all ICs an infeasible OCP for the $N_1$ cMPC was encountered. In all cases, this was due to overstepping the lower saturation bound on $x_2$. The spacecraft reached high velocities and was not able to decelerate in time to avoid constraint violation due to the short horizon of the cMPC . All other controllers successfully performed the docking maneuver for all ICs. It should be noted that, apart for $t_{comp}$, results for $N_2$ and $N_3$ were almost identical, with only slight differences in $u_{cost}$. The difference in time required to reach the final reference for the  the three controllers: RGMPC, cMPC $N_2,\;N_3$ was never longer than 1 sec.

Figure \ref{fig:cMPC_tcomp_av} shows statistics of $ u_{cost}$ and $t_{comp,av}$ for RGMPC and for cMPC with horizons $N_2$, $N_3$. With respect to the input cost, all controllers perform similarly, a slightly lower cost for the RGMPC is observed: for the $x_2 = 50\;[\text m]$ set medians are 0.7351 and 0.781 for the RGMPC and the two cMPC, respectively. In contrast, $t_{comp,av}$ (lower figure) differs substantially between each controller. The RGMPC has $t_{comp,av}$ that are one order of magnitude smaller than the cMPC with $N_2$ and almost two orders smaller than the cMPC with $N_3$. For this example, by looking at the $t_{comp}$ of the cMPC with $N_1$ (not presented here), it was assessed that the difference in $t_{comp,av}$ came from the difference in the MPC prediction horizon lengths and not so much from the additional constraints in the OCPs of the cMPC.
\begin{figure}[ht]
\centering
\includegraphics[scale=0.5]{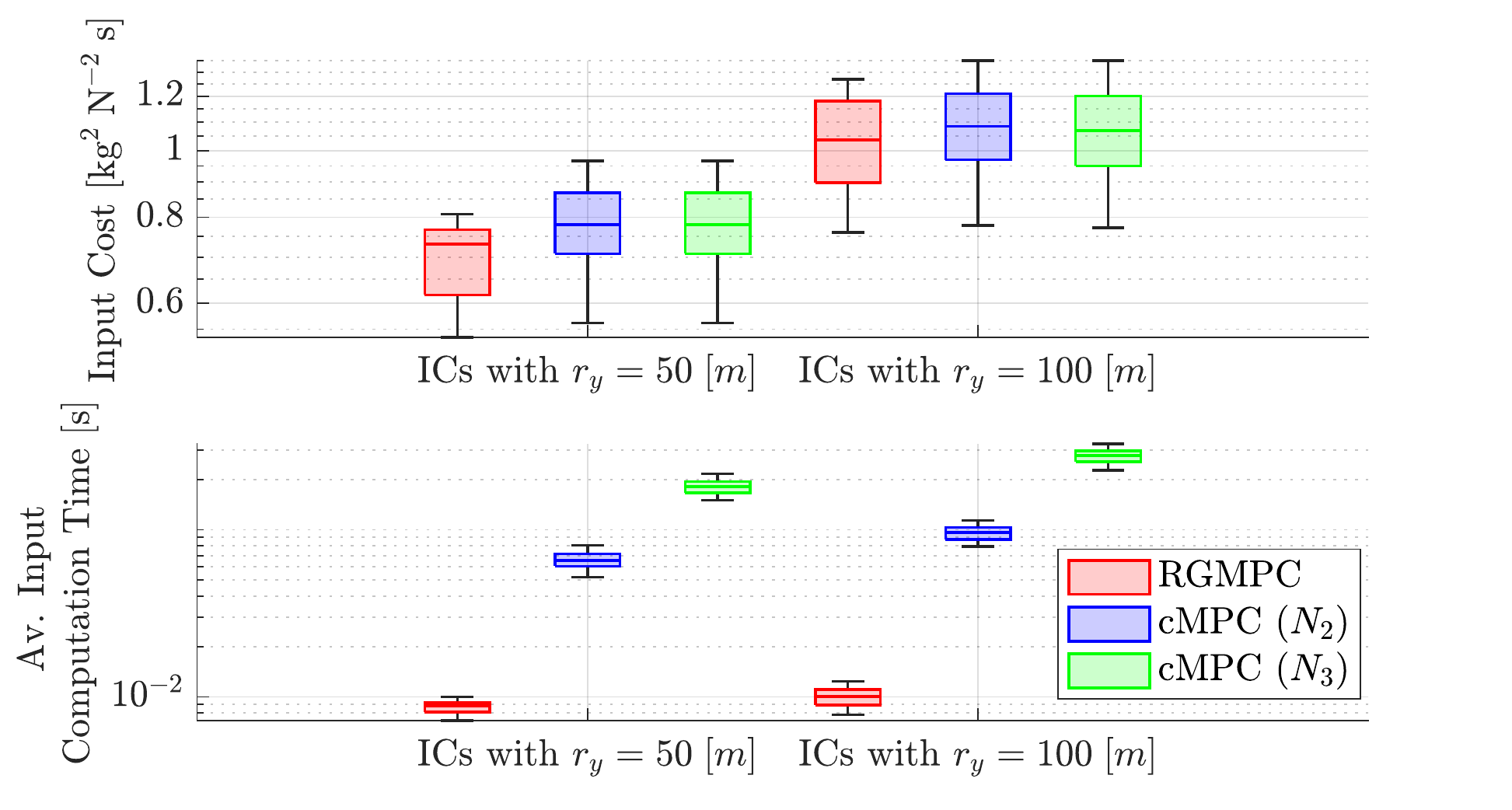}
\caption{Statistical values of $u_{cost}$ (upper) and $t_{comp,av}$ (lower) for the three controllers RGMPC, cMPC ($N_2$) and cMPC ($N_3$) and the two sets of IC: starting at 50 and 100 meters in track, respectively. The statistical values are : median, first and third quartile, min-max values.}
\label{fig:cMPC_tcomp_av}
\end{figure}

\section{Conclusion}\label{sec:Conclusion}
An input constrained Linear Quadratic MPC can be augmented by a variant of an incremental reference governor (IRG) to avoid violations of (possibly nonlinear) state constraints and nonlinear control constraints.  The proposed scheme is designed to avoid MPC optimization at every time instant over the IRG prediction horizon by relying on the previously computed MPC input sequence padded with the saturated LQR. Finite-time convergence properties of the modified IRG reference command to a strictly steady-state constraint admissible reference command have been established.  Simulation results demonstrate computational advantages of the proposed scheme over both input and state constrained MPC and performance advantages over a saturated LQR controller augmented with the IRG. 

\section{Acknowledgments}
We thank Dominic Liao-McPherson for providing the implementation of the dual active set solver used in Section \ref{sec:res_comp_cMPC}



\end{document}